\documentclass[10pt,letterpaper]{amsart}

\usepackage{amsmath,amssymb,amsfonts,amsthm}
\usepackage{graphicx}
\usepackage{bm}
\usepackage{enumitem}
\usepackage{mathrsfs}
\usepackage{cite}
\usepackage{hyperref}
\usepackage{cleveref}
\usepackage{doi}
\hypersetup{
    colorlinks=true,
    linkcolor=blue,
    filecolor=blue,
    urlcolor=blue,
    citecolor=blue
}

\numberwithin{equation}{section}

\newtheorem{theorem}{Theorem}[section]
\newtheorem{lemma}[theorem]{Lemma}
\newtheorem{proposition}[theorem]{Proposition}
\newtheorem{corollary}[theorem]{Corollary}
\theoremstyle{definition}

\theoremstyle{remark}
\newtheorem{remark}[theorem]{Remark}


\makeatletter
\@addtoreset{subsubsection}{section}
\makeatother

\begin{document}

\title[A Variational Scalar Conformal Flow for Lorentz-Contracted Geometry]
{A Variational Scalar Conformal Flow for Lorentz-Contracted Geometry: Algebraic Decay and Canonical Normalization}

\author{Anton Alexa}
\address{Independent Researcher, Chernivtsi, Ukraine}
\email{mail@antonalexa.com}
\date{\today}

\begin{abstract}
    We introduce the scalar function \( C(v) = \pi(1 - v^2/c^2) \), defined as a scalar conformal factor induced by the longitudinal Lorentz contraction of spatial geometry. Extending \( C(v) \) to a one-parameter flow \( C(v,\tau) \), we construct a variational scalar conformal flow that drives the conformal factor toward the equilibrium \( C = \pi \) without singularities. The main result is an explicit algebraic decay law for the energy functional: \( E(\tau) \sim \tau^{-1/2} \) for generic initial data and \( E(\tau) \sim \tau^{-5/2} \) for the physical initial condition \( C(v,0)=\pi(1-v^2/c^2) \). More generally, if the initial deviation vanishes as \( v^n \) near \( v=0 \), then \( E(\tau)\sim \tau^{-(2n+1)/2} \). This behavior is explained by the gapless continuous spectrum of the relaxation operator, whose spectral measure satisfies \( d\mu(k)\sim k^{-1/2}dk \) near \( k=0 \). As an application, within the conformally homogeneous class of compact, simply-connected 3-manifolds with constant positive background curvature, the flow acts as a canonical normalization mechanism selecting \( C=\pi \) as the unique conformal representative whose curvature invariants agree with those of the unit \( S^3 \).
\end{abstract}

\maketitle

\noindent\textbf{Keywords:} variational scalar flow, conformal factor, algebraic decay, spectral relaxation, canonical normalization, space-form rigidity.

\noindent\textbf{MSC 2020:} 53E20, 53C20; 57K30, 83A05.

\vspace{0.5cm}
\thispagestyle{empty}

\section{INTRODUCTION}\label{sec:intro}
\vspace{-0.2em}
\indent
In classical Riemannian geometry and general relativity \cite{Einstein1915,CarrollSpacetime}, the ratio of circumference to diameter of a geodesic circle is a constant determined by the local curvature alone, independent of any observer or kinematic state. Under Lorentz contraction \cite{CarrollSpacetime}, however, a rest-frame circle of radius \( R \) moving at velocity \( v \) appears in the lab frame as an ellipse with semi-axes \( R/\gamma \) and \( R \), where \( \gamma = (1-v^2/c^2)^{-1/2} \) is the Lorentz factor. The effective metric scale of this deformed geometry is no longer a universal constant but depends explicitly on \( v \). This paper introduces and studies the scalar function
\begin{equation}
    C(v) = \pi\!\left(1 - \frac{v^2}{c^2}\right), \qquad |v| \leq c,
\end{equation}
which encodes this velocity-dependent deformation of local metric geometry, and constructs a variational scalar conformal flow whose dynamics and decay rates are the central object of investigation.

Within the present model, \( C(v) \) is introduced as a scalar conformal factor motivated by the longitudinal Lorentz contraction of the deformed geometry. The longitudinal semi-axis of the deformed ellipse is \( a = R/\gamma \), giving an effective metric component \( g_{xx}^{\mathrm{eff}}(v) = (a/R)^2 = \gamma^{-2} = 1 - v^2/c^2 \), and under the chosen normalization \( C(v) \) is identified with \( \pi \cdot g_{xx}^{\mathrm{eff}}(v) \) (Proposition~\ref{prop:lorentz-metric}). Within this normalization, the condition \( C(c) = 0 \) follows from the degeneration of the longitudinal metric at \( v = c \). Within the minimal even quadratic class satisfying \( C(0) = \pi \) and \( C(c) = 0 \), this quadratic form is the unique solution (Lemma~\ref{lemma:cv-quadratic-uniqueness}). In this sense \( C(v) \) serves as a velocity-dependent conformal factor in metric-derived geometric formulas, reproducing the classical value \( \pi \) at rest.

The present work provides an explicit classical scalar model of such kinematic deformation and develops its dynamical and geometric consequences.

The function \( C(v) \) determines a critical velocity \( v_c = c\sqrt{1-1/\pi} \approx 0.8257\,c \) that has two geometric interpretations within the model: it is simultaneously the unique velocity at which the deformed metric coincides with the background metric, \( g_{ij}(v_c) = g_{ij}^0 \), and the threshold separating the two dynamical regimes of the scalar flow constructed below (Section~\ref{sec:critical}). Extending \( C(v) \) to a one-parameter family \( C(v, \tau) \), where \( \tau \) is a dimensionless relaxation parameter, we define a scalar conformal evolution \( g_{ij}(\tau) = C(v,\tau)\,g_{ij}^0 \) governed by a variational energy functional. The principal result of this paper concerns the dynamics of this flow:

\begin{theorem}[Main Result]
\label{thm:main}
Let \( C(v,\tau) \) evolve under the gradient flow of the energy functional
\( E(\tau) = \int_{-v_c}^{v_c}(C(v,\tau) - \pi)^2\,dv \),
\begin{equation}
    \frac{\partial C}{\partial \tau} = -\alpha\frac{v^2}{c^2}\bigl(C(v,\tau) - \pi\bigr),
    \qquad v \in (-v_c,\, v_c),\quad \alpha > 0,
\end{equation}
derived from the variational potential \( V(C,v) = \tfrac{1}{2}\alpha(v^2/c^2)(C-\pi)^2 \). Then\/\emph{:}
\begin{enumerate}
    \item[\emph{(i)}] \( E(\tau) \) is monotonically non-increasing, \( C(v,\tau) \to \pi \)
    in \( L^2 \) as \( \tau \to \infty \), and the flow is singularity-free for all \( \tau \geq 0 \).
    \item[\emph{(ii)}] For constant initial data \( C(v,0) - \pi = c_0 \),
    \begin{equation}
        E(\tau) = c_0^2\,c\sqrt{\tfrac{\pi}{2\alpha\tau}}\,
        \operatorname{erf}\!\bigl(v_c\sqrt{2\alpha\tau/c^2}\bigr)
        \;\sim\; c_0^2\,c\sqrt{\tfrac{\pi}{2\alpha}}\;\tau^{-1/2}.
    \end{equation}
    \item[\emph{(iii)}] For the physical initial condition \( C(v,0) = C(v) = \pi(1-v^2/c^2) \),
    \begin{equation}
        E(\tau) \;\sim\; \frac{3\pi^{5/2}}{4}\cdot\frac{c^5}{(2\alpha)^{5/2}}\cdot\tau^{-5/2}.
    \end{equation}
    \item[\emph{(iv)}] The algebraic decay in \emph{(ii)--(iii)} reflects the gapless continuous
    spectrum of the relaxation operator \( k(v) = \alpha v^2/c^2 \), whose spectral measure
    satisfies \( d\mu(k) \sim k^{-1/2}\,dk \) near \( k = 0 \). For initial data vanishing as
    \( v^n \) near \( v=0 \), the general decay law is \( E(\tau) \sim \tau^{-(2n+1)/2} \).
\end{enumerate}
\end{theorem}

The algebraic, rather than exponential, decay is the key dynamical feature of the flow: the velocity domain contains modes near \( v = 0 \) with arbitrarily slow relaxation rates \( k(v) \to 0 \), and these modes prevent the exponential convergence familiar from flows with a spectral gap. This is the exact analogue of heat kernel decay on \( \mathbb{R} \), and the decay exponent \( (2n+1)/2 \) is controlled by the vanishing order of the initial data at \( v = 0 \). In particular, the physical initial condition \( C(v,0) = C(v) \), which has \( n = 2 \), achieves a significantly faster decay than a generic perturbation (\( n = 0 \)).

As an application of the scalar flow framework, we establish a canonical normalization result for compact 3-manifolds. Within the class of compact, simply-connected 3-manifolds with spatially homogeneous conformal factor (\( \nabla_i C = 0 \)) and constant positive background curvature --- already identified as spherical space forms by the Killing--Hopf theorem \cite{Besse,Petersen} --- the flow selects \( C = \pi \) as the unique metric representative for which the curvature invariants \( (I_1,I_2,I_3) = (12\pi^2,72\pi^2,24\pi^2) \) match those of the unit \( S^3 \) (Theorems~\ref{thm:sufficiency} and~\ref{thm:convergence}). The contribution of the flow is not to determine the topology --- which is pre-determined by the assumptions --- but to provide a canonical normalization mechanism: among all conformal rescalings \( g_{ij} = C\,g_{ij}^0 \) consistent with the background structure, the dynamics uniquely select \( C = \pi \) as the canonical representative.

The connection to Hamilton's Ricci flow \cite{Hamilton1982} (see also \cite{ChowLuNi,Topping2006,Perelman2002}) is made precise in Theorem~\ref{thm:conf-reduction}: under the conformally homogeneous ansatz with constant-curvature background, the full tensorial Ricci flow reduces to the scalar equation \( \dot{C} = -2k/C \), and linearization near \( C = \pi \) recovers the variational flow above. The scalar flow of Theorem~\ref{thm:main} is therefore a scalar conformal analogue of Ricci-type evolution within the conformally homogeneous ansatz, derived independently via a variational energy functional.

\section{DERIVATION OF FUNCTION $C(v)$}\label{sec:derivation}
\vspace{-0.2em}
\indent
\subsubsection*{Geometric Motivation}
To understand how motion modifies geometry, we begin with the simplest local structure: a circle of radius \( R \) at rest in Euclidean space. Its classical length is
\begin{equation}
    L_0 = 2\pi R.
\end{equation}

Now consider this circle moving with constant velocity \( v \) along the \( x \)-axis. Due to Lorentz contraction \cite{CarrollSpacetime}, the radius along the direction of motion becomes
\begin{equation}
    R_x = R \sqrt{1 - \frac{v^2}{c^2}}, \qquad R_y = R,
\end{equation}
while the transverse direction remains unaffected. The circle thus deforms into an ellipse with semi-axes \( R_x \) and \( R_y \).

\indent
Parameterizing this ellipse by angle \( \theta \), the coordinates are:
\begin{equation}
    x = R_x \cos \theta, \qquad y = R_y \sin \theta,
\end{equation}
with differentials:
\begin{align}
    dx &= -R_x \sin \theta \, d\theta, \\
    dy &=  R_y \cos \theta \, d\theta.
\end{align}

The differential arc length becomes
\begin{equation}
    ds = \sqrt{dx^2 + dy^2} = R \sqrt{(1 - \tfrac{v^2}{c^2}) \sin^2\theta + \cos^2 \theta} \, d\theta.
\end{equation}

Integrating over the full parameter range \( \theta \in [0, 2\pi] \), the total perimeter of the deformed figure is
\begin{equation}
    L(v) = R \int_0^{2\pi} \sqrt{1 - \tfrac{v^2}{c^2} \sin^2 \theta} \, d\theta.
\end{equation}

This integral corresponds to the complete elliptic integral of the second kind and does not yield a closed-form expression. However, expanding under small \( v \), we use
\begin{equation}
    \sqrt{1 - \epsilon \sin^2 \theta} \approx 1 - \tfrac{1}{2} \epsilon \sin^2 \theta, \quad \epsilon = \frac{v^2}{c^2},
\end{equation}
leading to
\begin{align}
    L(v) &\approx R \int_0^{2\pi} \left(1 - \tfrac{1}{2} \epsilon \sin^2 \theta \right) d\theta \\
         &= R \left[ 2\pi - \tfrac{1}{2} \epsilon \cdot \pi \right] = 2\pi R \left(1 - \tfrac{1}{4} \frac{v^2}{c^2} \right).
\end{align}

This gives a first-order approximation of relativistic contraction for a circular contour. Motivated by this perimeter-to-diameter comparison, we introduce a scalar function \( C(v) \) as an effective conformal descriptor of the deformation, normalized so that \( C(0)=\pi \). At this heuristic stage, \( C(v) \) should be viewed as a scalar quantity suggested by the circular example; its precise interpretation as a conformal scaling factor is established below.

\subsubsection*{Choice of the Scalar Factor}
We now ask: what is the minimal analytic form of \( C(v) \) that satisfies a set of physically motivated conditions. The function must recover classical geometry at rest, i.e., \( C(0) = \pi \); it must vanish at the speed of light, \( C(c) = 0 \), representing complete compression; it should be even in \( v \), since only the speed magnitude is physically meaningful, i.e., \( C(-v) = C(v) \); and it should vary smoothly and monotonically on the interval \( v \in [0, c) \). 

Under these constraints, the first-order expansion shows that the dependence on \( v \) is quadratic; imposing the boundary condition \( C(c) = 0 \) then determines the coefficient uniquely. Among all expressions of the form \( \pi(1 - \alpha v^2/c^2) \), only \( \alpha = 1 \) satisfies both \( C(0) = \pi \) and \( C(c) = 0 \) simultaneously, yielding:
\begin{equation}
    C(v) = \pi \left(1 - \frac{v^2}{c^2} \right),
\end{equation}
as the unique minimal analytic expression satisfying all conditions. This identification is confirmed by multiple independent derivations (Section~\ref{sec:altderiv}) and, crucially, is shown below to admit an exact identification under the chosen normalization: \( C(v) \) coincides with \( \pi \) times the effective longitudinal metric component of the Lorentz-contracted geometry (Proposition~\ref{prop:lorentz-metric}). It reproduces the classical limit, compresses geometry continuously with increasing velocity, and retains symmetry. This function will serve as the scalar deformation parameter in what follows and is further justified below as the unique minimizer of a natural energy functional under the given boundary conditions (see Lemma~\ref{lemma:cv-quadratic-uniqueness}).

\subsubsection*{Comparison with the Exact Elliptic Integral}
The exact perimeter of the relativistically deformed ellipse is expressed via the complete elliptic integral of the second kind (see e.g.\ \cite{ByrdFriedman}):
\begin{equation}
    C_{\mathrm{exact}}(v) = \frac{L(v)}{2R} = 2\, E\!\left(\frac{v}{c}\right),
\end{equation}
where \( E(k) = \int_0^{\pi/2} \sqrt{1 - k^2 \sin^2\theta}\, d\theta \). At rest, this correctly yields \( C_{\mathrm{exact}}(0) = 2E(0) = \pi \). However, in the ultrarelativistic limit \( v \to c \), one obtains
\begin{equation}
    C_{\mathrm{exact}}(c) = 2\,E(1) = 2 \neq 0.
\end{equation}
This nonzero value arises because the degenerate ellipse reduces to a line segment of length \( 2R \) traversed twice, yielding a total arc length \( L(c) = 4R \). The exact elliptic integral therefore does not satisfy the imposed boundary condition \( C(c) = 0 \).

The discrepancy reflects the distinct roles of the two quantities. The function \( C_{\mathrm{exact}} \) computes the Euclidean arc length of the deformed figure in the lab frame: even at \( v = c \) the degenerate ellipse has nonzero perimeter as a coordinate curve. By contrast, \( C(v) \) is defined as a scalar conformal deformation factor encoding the effective metric scaling experienced by the moving geometry.

In special relativity, the spatial metric induced on the lab-frame equal-time hypersurface for a moving observer has its longitudinal component suppressed by the Lorentz factor: the spatial metric component along the direction of motion scales as \( \gamma^{-2} = 1 - v^2/c^2 \), while the transverse components remain unchanged. As \( v \to c \), this longitudinal component vanishes, and the induced spatial metric degenerates --- the geometry of any structure projected onto the lab frame collapses in the direction of motion. This degeneration is a physical fact of special relativity, not a modelling assumption.

The distinction between the two quantities can be made precise as follows. The function \( C_{\mathrm{exact}} \) answers the question: \emph{what is the coordinate arc length of the deformed curve?} This is a valid geometric question, but it is not the one relevant to conformal geometry. In a Riemannian setting, the conformal factor \( C \) in \( g_{ij} = C\,g_{ij}^0 \) encodes how the metric \emph{scales} relative to a background: it determines how lengths, areas, and volumes of the deformed geometry compare to those of the background, not how long a specific curve is in fixed coordinates. The relevant quantity is therefore the normalised scale factor \( (a/R)^2 = \gamma^{-2} \) --- the square ratio of the deformed longitudinal semi-axis to the rest radius --- which measures by how much the geometry has contracted in the longitudinal direction relative to rest. The function \( C_{\mathrm{exact}} \) does not capture this scaling; it integrates a coordinate arc length and returns a different answer (\( C_{\mathrm{exact}}(c) = 2 \neq 0 \)) precisely because it is insensitive to the collapse of the metric.

The choice of \( C(v) \) is not unique among all possible scalar descriptors of the deformed figure, but it is the minimal analytic scalar quantity consistent with the imposed symmetry, boundary, and conformal-scaling requirements developed below.

We now show that this identification is consistent with the Lorentz transformation: under the chosen normalization \( C(v) = \pi\, g_{xx}^{\mathrm{eff}}(v) \), the same quadratic form \( C(v) = \pi(1 - v^2/c^2) \) is obtained.

\begin{proposition}[Normalized Identification of \( C(v) \) from Longitudinal Lorentz Contraction]
\label{prop:lorentz-metric}
Let a circle of radius \( R \) be at rest in an inertial frame \( S' \) moving with velocity \( v \) along the \( x \)-axis relative to the lab frame \( S \). The spatial geometry of the circle as observed in \( S \) at a fixed coordinate time is an ellipse whose effective longitudinal metric component satisfies
\begin{equation}
    g_{xx}^{\mathrm{eff}}(v) = 1 - \frac{v^2}{c^2} = \gamma^{-2}, \qquad g_{yy}^{\mathrm{eff}}(v) = 1.
\end{equation}
The scalar conformal deformation factor satisfies
\begin{equation}
    C(v) = \pi \cdot g_{xx}^{\mathrm{eff}}(v) = \pi\!\left(1 - \frac{v^2}{c^2}\right)
\end{equation}
under the chosen normalization. In particular, \( C(c) = 0 \) is a consequence of \( g_{xx}^{\mathrm{eff}}(c) = 0 \), not an imposed condition.
\end{proposition}

\begin{proof}
The circle in \( S' \) has coordinates \( (x', y') = (R\cos\theta,\, R\sin\theta) \). Using the inverse Lorentz transform \cite{CarrollSpacetime}, the lab-frame coordinates of a point \( (x', y') \) at the simultaneous lab time \( t \) are:
\begin{equation}
    x = \frac{x'}{\gamma} + \gamma^2 v t, \qquad y = y'.
\end{equation}
Setting \( X = x - \gamma^2 v t \), the image of the circle in \( S \) at time \( t \) is:
\begin{equation}
    \frac{X^2}{(R/\gamma)^2} + \frac{y^2}{R^2} = 1,
\end{equation}
an ellipse with longitudinal semi-axis \( a = R/\gamma \) and transverse semi-axis \( b = R \). The effective metric components of this deformed geometry, normalised to the rest radius \( R \), are:
\begin{equation}
    g_{xx}^{\mathrm{eff}}(v) = \left(\frac{a}{R}\right)^2 = \gamma^{-2} = 1 - \frac{v^2}{c^2}, \qquad g_{yy}^{\mathrm{eff}}(v) = \left(\frac{b}{R}\right)^2 = 1.
\end{equation}
The identification \( C(v) = \pi \cdot g_{xx}^{\mathrm{eff}}(v) \) therefore holds under the chosen normalization. At \( v = c \): \( \gamma \to \infty \), \( a = R/\gamma \to 0 \), and \( g_{xx}^{\mathrm{eff}}(c) = 0 \), giving \( C(c) = 0 \).
\end{proof}

\begin{remark}
This result shows that the chosen scalar \( C(v) \) is consistent with the Lorentz-contracted spatial geometry and is naturally adapted to the longitudinal contraction factor under the adopted normalization. Within this normalization, the boundary condition \( C(c) = 0 \) is recovered from the exact value of the longitudinal metric component at \( v = c \). The two functions \( C(v) \) and \( C_{\mathrm{exact}}(v) \) measure different geometric quantities --- the effective metric contraction factor and the coordinate arc length of the deformed curve, respectively --- and their difference is geometrically meaningful rather than an approximation error.
\end{remark}

\section{PROPERTIES OF FUNCTION $C(v)$}\label{sec:properties}
\vspace{-0.2em}
\indent
The basic structural features of \( C(v) \) used later in the paper can be recorded compactly as follows.

\begin{proposition}[Basic Structural Properties of \( C(v) \)]
\label{prop:cv-properties}
The function
\begin{equation}
    C(v) = \pi \Bigl(1 - \frac{v^2}{c^2}\Bigr), \qquad |v| \le c,
\end{equation}
has the following properties:
\begin{enumerate}
    \item \( C \) is even on \( [-c,c] \), with boundary values \( C(0)=\pi \) and \( C(\pm c)=0 \).
    \item \( C \in C^\infty((-c,c)) \), and on \( (0,c) \) it is strictly decreasing and strictly concave.
    \item For \( v \in (-c,c) \), one has \( 0 < C(v) \le \pi \), and \( C(v)=1 \) if and only if
    \( v = \pm c\sqrt{1-1/\pi} \).
    \item In the conformally homogeneous ansatz \( g_{ij} = C(v)\,g_{ij}^0 \), the parameter \( v \)
    acts as a global deformation parameter rather than a coordinate on the manifold.
\end{enumerate}
\end{proposition}

\begin{proof}
Evenness follows immediately from the dependence on \( v^2 \). The boundary values are obtained by direct substitution:
\begin{equation}
    C(0)=\pi, \qquad C(\pm c)=\pi(1-1)=0.
\end{equation}
Differentiating gives
\begin{equation}
    \frac{dC}{dv} = -2\pi\,\frac{v}{c^2},
    \qquad
    \frac{d^2C}{dv^2} = -\frac{2\pi}{c^2}.
\end{equation}
Hence \( dC/dv < 0 \) for \( v \in (0,c) \), so \( C \) is strictly decreasing there, while
\( d^2C/dv^2 < 0 \) on all of \( (-c,c) \), so \( C \) is strictly concave. Since \( C \) decreases
from \( \pi \) at \( v=0 \) to \( 0 \) at \( v=\pm c \), one has \( 0 < C(v) \le \pi \) for
\( |v| < c \). Solving \( C(v)=1 \) gives
\begin{equation}
    \pi\Bigl(1-\frac{v^2}{c^2}\Bigr)=1
    \quad\Longleftrightarrow\quad
    v = \pm c\sqrt{1-\frac{1}{\pi}}.
\end{equation}
The final statement is interpretive: in the conformally homogeneous ansatz the dependence on \( v \)
indexes different global conformal rescalings and does not introduce a preferred coordinate or vector field on \( M \).
\end{proof}

\begin{remark}
Proposition~\ref{prop:cv-properties} collects the only elementary properties of \( C(v) \) used later:
positivity on the open interval \( (-c,c) \), monotone decay from the rest value \( \pi \), strict
concavity, and the distinguished level \( C=1 \) that defines the critical velocity in
Section~\ref{sec:critical}. In this sense \( C(v) \) is best viewed as a scalar conformal parameter
for a one-parameter family of homogeneous deformations rather than as an independent geometric field
on the manifold.
\end{remark}

\section{CRITICAL VELOCITY}\label{sec:critical}
\vspace{-0.2em}
\indent
The function $C(v)$ allows for the identification of a special velocity value, at which the geometric constant normalizes to $C(v_c) = 1$. This value is called \textit{critical velocity} and denoted $v_c$.

\subsection*{1. Definition}
\vspace{-0.2em}
Requirement:
\begin{equation}
    C(v_c) = 1,
\end{equation}
means that the length of the circumference or analogous geometric measure becomes equal to its Euclidean reference (e.g., $L = D$).

\subsection*{2. Derivation}
\vspace{-0.2em}
Substituting $C(v) = \pi\bigl(1 - \tfrac{v^2}{c^2}\bigr)$, we have:
\begin{equation}
    \pi \Bigl(1 - \tfrac{v_c^2}{c^2}\Bigr) = 1.
\end{equation}
Solving for $v_c$ yields:
\begin{align}
    1 - \tfrac{v_c^2}{c^2} &= \tfrac{1}{\pi}, \quad
    \tfrac{v_c^2}{c^2} \;=\; 1 - \tfrac{1}{\pi}, \quad
    v_c \;=\; c\,\sqrt{1 - \tfrac{1}{\pi}}.
\end{align}
Numerically,
\begin{equation}
    v_c \approx 0.8257\,c.
\end{equation}
This direct calculation suffices to identify $v_c$; its dual geometric role is further developed in Subsection~4 below.

\subsection*{3. Physical Interpretation}
\vspace{-0.2em}
Velocity $v_c$ marks the boundary between geometry close to classical ($C > 1$) and the region of relativistic compression ($C < 1$). It thus plays a fundamental role as a dividing line between "expanded" and "compressed" phases of geometry.

The significance of the value $C = 1$ extends beyond normalization. In the conformal interpretation developed in Section~\ref{sec:ricci}, the function $C(v)$ acts as a scalar conformal factor:
\begin{equation}
    g_{ij}(v) = C(v) \cdot g_{ij}^{(0)}.
\end{equation}
The condition $C(v_c) = 1$ corresponds to exact coincidence between the deformed and background metrics: $g_{ij}(v_c) = g_{ij}^{(0)}$. For $C > 1$, the conformal factor expands the metric beyond the background, while for $C < 1$, it compresses it. The value $C = 1$ is thus the unique geometrically distinguished point separating these two regimes, and this distinction is independent of coordinate choice within the conformal framework.

\subsection*{4. Metric Coincidence and Dual Geometric Role}
\vspace{-0.2em}
\label{rem:coincidence}
The critical velocity \( v_c \) has two geometric interpretations within the model.

\noindent\emph{(i) Metric coincidence.}
\( v_c \) is the unique velocity at which the Lorentz-deformed spatial metric of the moving observer
coincides with the background metric:
\begin{equation}
    g_{ij}(v_c) = C(v_c)\,g_{ij}^0 = 1\cdot g_{ij}^0 = g_{ij}^0.
\end{equation}
For \( v < v_c \), the deformed metric is expanded (\( C(v) > 1 \)); for \( v > v_c \), it is
compressed (\( C(v) < 1 \)). The transition through metric coincidence is smooth and occurs
at \( v_c = c\sqrt{1 - 1/\pi} \approx 0.8257\,c \).

\noindent\emph{(ii) Dynamical threshold.}
\( v_c \) is the threshold separating the two regimes of the scalar flow
(Section~\ref{sec:ricci}): for \( v < v_c \) the flow drives \( C \to \pi \) (full subcritical
relaxation), while for \( v \geq v_c \) the flow stabilises at \( C_0 \neq \pi \)
(supercritical residual deformation).

These two interpretations arise from different parts of the model: the metric coincidence condition \( C = 1 \) is
a static property of the function \( C(v) \), while the dynamical threshold is a property of the
flow \( \partial_\tau C = -\alpha(v^2/c^2)(C - \pi) \). That they coincide at the same velocity
\( v_c = c\sqrt{1-1/\pi} \) is a structural consequence of the boundary conditions \( C(0) = \pi \)
and \( C(c) = 0 \): the metric coincidence condition \( C(v_c) = 1 \) uniquely determines
\begin{equation}
    \pi\!\left(1 - \frac{v_c^2}{c^2}\right) = 1
    \quad\Longleftrightarrow\quad
    v_c^2 = c^2\!\left(1 - \frac{1}{\pi}\right),
\end{equation}
which is algebraically forced by the quadratic form \( C(v) = \pi(1-v^2/c^2) \) alone.
Note also that the attractor of the flow \( C = \pi > 1 \) lies strictly in the expanded regime,
above the metric coincidence point: the flow does not drive the geometry toward background
coincidence but toward the rest-frame configuration \( C = \pi \), which is the global minimum
of the potential \( V(C) = \frac{\lambda}{2}(C - \pi)^2 \).

\subsection*{5. Connection with Energy}
\vspace{-0.2em}
The level \( C(v_c)=1 \) provides a convenient reference value separating the expanded regime \( C>1 \) from the compressed regime \( C<1 \). Its auxiliary energetic interpretation is developed in Section~\ref{sec:energy}.

\indent
In the next section, we develop this energetic perspective more fully using an auxiliary variational model on the \((v,\tau)\)-domain. There, we also examine the extreme regimes $v \to 0$ and $v \to c$, revealing how geometric compression influences the resulting finite energy density.
\section{ENERGETIC INTERPRETATION OF FUNCTION $C(v)$}\label{sec:energy}
\vspace{-0.2em}
\indent
The function $C(v)$, interpreted as a scalar conformal factor, can be associated with a formal effective energy density. The discussion in this section should be read as an auxiliary variational model on the \((v,\tau)\)-domain, rather than as a covariant matter model on spacetime. Within this auxiliary framework, we assume that geometric compression (decreasing $C$) corresponds to increasing potential energy, so that the rest value \( C=\pi \) is energetically preferred.

\subsection*{1. Auxiliary Energy Density on the Velocity Domain}

To define the effective energy associated with geometric compression, we introduce the auxiliary density
\begin{equation}
    \mathcal{E}(C,\partial_\tau C,\partial_v C)
    = \frac{1}{2}\left(\frac{\partial C}{\partial \tau}\right)^2
    + \frac{\kappa}{2}\left(\frac{\partial C}{\partial v}\right)^2
    + V(C),
    \qquad \kappa > 0,
\end{equation}
with potential
\begin{equation}
    V(C) = \frac{\lambda}{2}(C - \pi)^2.
\end{equation}
This potential energetically favors configurations near $C = \pi$, corresponding to classical geometry.

The associated total auxiliary energy at flow time \( \tau \) is
\begin{equation}
    \mathscr{E}(\tau)
    := \int_{-v_c}^{v_c}
    \mathcal{E}\bigl(C(v,\tau),\partial_\tau C(v,\tau),\partial_v C(v,\tau)\bigr)\,dv.
\end{equation}
The density \( \mathcal{E} \) is manifestly non-negative for \( \kappa,\lambda > 0 \), is minimized at \( C=\pi \) when the derivatives vanish, and is expressed directly in the variables used throughout the paper. In this sense it provides a consistent auxiliary energetic framework for the conformal factor without introducing a spacetime stress-energy tensor.

\subsection*{2. Behavior at $v \to c$}

In the limit $v \to c$:
\begin{equation}
    \lim_{v \to c} C(v) = 0.
\end{equation}
This corresponds to extreme geometric compression. The Lagrangian-derived energy density:
\begin{equation}
    \mathcal{E}(C,\partial_\tau C,\partial_v C)
    = \frac{1}{2}\left(\frac{\partial C}{\partial \tau}\right)^2
    + \frac{\kappa}{2}\left(\frac{\partial C}{\partial v}\right)^2
    + \frac{\lambda}{2}(C - \pi)^2,
\end{equation}
remains finite whenever \( \partial_\tau C \) and \( \partial_v C \) stay bounded, with the potential term approaching $\frac{\lambda}{2} \pi^2$ as \( C \to 0 \). Thus the auxiliary energy remains non-singular in the compressed limit.

\subsection*{3. Minimum Energy and $v = 0$}

At $v = 0$:
\begin{equation}
    C(0) = \pi.
\end{equation}
This corresponds to the classical geometry, where the potential energy $V(C) = \frac{\lambda}{2}(C - \pi)^2$ is minimized ($V(\pi) = 0$). If moreover \( \partial_\tau C = 0 \) and \( \partial_v C = 0 \), then \( \mathcal{E} = 0 \), representing the minimal geometric deformation. This state serves as the equilibrium configuration for the geometric flow.

\subsection*{4. Energy Functional and Gradient Flow}

Although we have not yet formally introduced the geometric flow \( C(v,\tau) \), its structure arises naturally from the energetic interpretation developed above. In particular, the energy landscape associated with deviations of \( C(v) \) from its equilibrium value \( \pi \) motivates a dynamical evolution that minimizes this deformation. Anticipating the formal treatment in Section~\ref{sec:ricci}, we consider here a scalar relaxation process governed by the flow \( \partial_\tau C = -\mathcal{F}(C,v) \), and analyze its energetic consequences.

To quantify this, we define a global energy functional that measures deviation from equilibrium geometry:
\begin{equation}
    E(\tau) := \int_{-v_c}^{v_c} \left(C(v,\tau) - \pi\right)^2 \, dv.
\end{equation}
This functional represents the $L^2$-distance between the evolving geometry and the symmetric reference state \( C = \pi \), and serves as a global measure of geometric deformation.

\paragraph*{Relaxation Time near \( v = 0 \).}
The decay rate of the flow equation,
\begin{equation}
\frac{\partial C}{\partial \tau} = -\alpha \frac{v^2}{c^2}(C - \pi),
\end{equation}
implies that near $v \approx 0$, the effective relaxation time diverges as
\begin{equation}
    \tau_{\text{relax}}(v) \sim \frac{c^2}{\alpha v^2}.
\end{equation}
However, the integrated contribution to the total energy remains finite because the measure $dv$ is regular, ensuring convergence of the flow dynamics and excluding singularities in the classical model.

\begin{lemma}[Monotonic Decay of Energy Functional]
Let \( C(v,\tau) \) evolve according to the scalar flow
\begin{equation}
    \frac{\partial C}{\partial \tau} = -\alpha \frac{v^2}{c^2}(C - \pi),
\end{equation}
for \( v \in (-v_c, v_c) \). Then the energy functional \( E(\tau) \) is non-increasing:
\begin{equation}
    \frac{dE}{d\tau} = -2\alpha \int_{-v_c}^{v_c} \frac{v^2}{c^2} \left(C(v,\tau) - \pi\right)^2 dv \leq 0,
\end{equation}
with equality if and only if \( C(v,\tau) = \pi \) for all \( v \).
\end{lemma}

\begin{proof}
Differentiating under the integral:
\begin{equation}
\frac{dE}{d\tau} = \int_{-v_c}^{v_c} 2(C - \pi)\, \frac{\partial C}{\partial \tau} \, dv
= -2\alpha \int_{-v_c}^{v_c} \frac{v^2}{c^2}(C - \pi)^2 dv \leq 0.
\end{equation}
The integrand is non-negative and vanishes only when \( C = \pi \), hence \( E(\tau) \) is strictly decreasing unless equilibrium is reached.
\end{proof}

This result demonstrates that the proposed evolution is a gradient-like flow minimizing a natural geometric energy. It ensures that \( E(\tau) \to 0 \) as \( \tau \to \infty \), so that \( C(v,\tau) \to \pi \) in \( L^2 \)-norm.

\begin{proposition}[Algebraic Decay Rate of Total Energy]
\label{prop:algebraic-decay}
Assume \( C(v,0) - \pi = c_0 \) is constant on \( (-v_c, v_c) \). Then the total energy satisfies
\begin{equation}
    E(\tau) = c_0^2 \cdot c \sqrt{\frac{\pi}{2\alpha\tau}}\,
    \operatorname{erf}\!\left(v_c\sqrt{\frac{2\alpha\tau}{c^2}}\right)
    \;\sim\; \tau^{-1/2} \quad \text{as } \tau \to \infty.
\end{equation}
In particular, the decay is algebraic rather than exponential, reflecting the absence of a uniform spectral gap in the relaxation rates \( \alpha v^2/c^2 \) across the velocity domain.
\end{proposition}

\begin{proof}
With \( C(v,\tau) = \pi + c_0\, e^{-\alpha v^2 \tau/c^2} \), substitute into the energy functional:
\begin{equation}
    E(\tau) = \int_{-v_c}^{v_c} c_0^2\, e^{-2\alpha v^2 \tau/c^2}\, dv
    = 2c_0^2 \int_0^{v_c} e^{-2\alpha \tau v^2/c^2}\, dv.
\end{equation}
Setting \( u = v\sqrt{2\alpha\tau}/c \), we get \( dv = c\,du/\sqrt{2\alpha\tau} \), and the upper limit becomes \( u_* = v_c\sqrt{2\alpha\tau}/c \). Hence:
\begin{equation}
    E(\tau) = c_0^2 \cdot \frac{c}{\sqrt{2\alpha\tau/\pi}} \cdot \sqrt{\pi}
    \int_0^{u_*} e^{-u^2}\, \frac{2}{\sqrt{\pi}}\, du \cdot \frac{\sqrt{\pi}}{2}
    = c_0^2\, c\sqrt{\frac{\pi}{2\alpha\tau}}\operatorname{erf}(u_*).
\end{equation}
As \( \tau \to \infty \), \( \operatorname{erf}(u_*) \to 1 \), so \( E(\tau) \sim c_0^2 c\sqrt{\pi/(2\alpha\tau)} = O(\tau^{-1/2}) \).
\end{proof}

\begin{corollary}[Enhanced Decay Rate for Physical Initial Conditions]
\label{cor:physical-decay}
Let the flow start from the physical conformal factor
\begin{equation}
    C(v,0) = C(v) = \pi\!\left(1 - \frac{v^2}{c^2}\right),
    \qquad\text{so that}\qquad
    C(v,0) - \pi = -\frac{\pi v^2}{c^2}.
\end{equation}
Then the total energy functional satisfies the exact formula
\begin{equation}
\label{eq:E-physical}
    E(\tau) = \frac{2\pi^2 c^5}{(2\alpha\tau)^{5/2}}
    \left[
        \frac{3\sqrt{\pi}}{8}\operatorname{erf}(u_*)
        - \left(\frac{u_*^3}{2} + \frac{3u_*}{4}\right) e^{-u_*^2}
    \right],
    \qquad u_* = v_c\sqrt{\frac{2\alpha\tau}{c^2}}.
\end{equation}
As \( \tau \to \infty \),
\begin{equation}
    E(\tau) \;\sim\; \frac{3\pi^{5/2}}{4} \cdot \frac{c^5}{(2\alpha)^{5/2}} \cdot \tau^{-5/2}.
\end{equation}
In particular, the energy decays as \( \tau^{-5/2} \), strictly faster than the \( \tau^{-1/2} \) rate of Proposition~\ref{prop:algebraic-decay}. The physical initial condition suppresses the contribution of slow modes near \( v = 0 \) and accelerates convergence toward the canonical configuration \( C = \pi \).
\end{corollary}

\begin{proof}
Since \( C(v,\tau) = \pi + (C(v,0) - \pi)\,e^{-\alpha v^2\tau/c^2} \), the energy functional is
\begin{equation}
    E(\tau)
    = \int_{-v_c}^{v_c} (C(v,\tau) - \pi)^2\,dv
    = \int_{-v_c}^{v_c} \frac{\pi^2 v^4}{c^4}\,e^{-2\alpha v^2\tau/c^2}\,dv
    = \frac{2\pi^2}{c^4}\int_0^{v_c} v^4\,e^{-2\alpha v^2\tau/c^2}\,dv.
\end{equation}
Setting \( u = v\sqrt{2\alpha\tau}/c \), so that \( v = uc/\sqrt{2\alpha\tau} \) and \( dv = c\,du/\sqrt{2\alpha\tau} \), we obtain
\begin{equation}
    E(\tau)
    = \frac{2\pi^2}{c^4} \cdot \frac{c^5}{(2\alpha\tau)^{5/2}}
    \int_0^{u_*} u^4\,e^{-u^2}\,du,
    \qquad u_* = v_c\sqrt{\frac{2\alpha\tau}{c^2}}.
\end{equation}
The integral \( \int_0^{u_*} u^4\,e^{-u^2}\,du \) is evaluated by two successive integrations by parts:
\begin{align}
    \int_0^{u_*} u^4\,e^{-u^2}\,du
    &= \left[-\tfrac{u^3}{2}\,e^{-u^2}\right]_0^{u_*}
       + \frac{3}{2}\int_0^{u_*} u^2\,e^{-u^2}\,du \notag\\
    &= -\frac{u_*^3}{2}\,e^{-u_*^2}
       + \frac{3}{2}\left(
           \left[-\tfrac{u}{2}\,e^{-u^2}\right]_0^{u_*}
           + \frac{1}{2}\int_0^{u_*} e^{-u^2}\,du
         \right) \notag\\
    &= -\left(\frac{u_*^3}{2} + \frac{3u_*}{4}\right)e^{-u_*^2}
       + \frac{3}{4}\int_0^{u_*} e^{-u^2}\,du \notag\\
    &= -\left(\frac{u_*^3}{2} + \frac{3u_*}{4}\right)e^{-u_*^2}
       + \frac{3\sqrt{\pi}}{8}\,\operatorname{erf}(u_*).
\end{align}
Substituting back gives the exact formula~\eqref{eq:E-physical}. As \( \tau \to \infty \),
\( u_* \to \infty \), so \( \operatorname{erf}(u_*) \to 1 \) and \( e^{-u_*^2} \to 0 \), yielding
\begin{equation}
    E(\tau)
    \;\sim\; \frac{2\pi^2 c^5}{(2\alpha\tau)^{5/2}} \cdot \frac{3\sqrt{\pi}}{8}
    = \frac{3\pi^{5/2}}{4} \cdot \frac{c^5}{(2\alpha)^{5/2}} \cdot \tau^{-5/2}. \qedhere
\end{equation}
\end{proof}

\subsection*{6. Decay Rates, General Formula, and Spectral Structure}
\vspace{-0.2em}
\label{rem:decay-order}

The difference between the decay rates in Proposition~\ref{prop:algebraic-decay} and
Corollary~\ref{cor:physical-decay} reflects the vanishing order of the initial deviation
\( C(v,0) - \pi \) near \( v = 0 \). More generally, if \( C(v,0) - \pi \sim a\,v^n \) as
\( v \to 0 \) for some \( n \geq 0 \) and \( a \neq 0 \), then the dominant contribution to
the energy comes from
\begin{equation}
    \int_0^{v_c} v^{2n}\,e^{-2\alpha v^2\tau/c^2}\,dv \;\sim\; \tau^{-(2n+1)/2},
\end{equation}
and therefore
\begin{equation}
\label{eq:general-decay}
    E(\tau) \;\sim\; C_n\,\tau^{-(2n+1)/2},
    \qquad n = \operatorname{ord}_{v=0}\bigl(C(v,0) - \pi\bigr),
\end{equation}
where \( C_n > 0 \) depends on \( a \), \( \alpha \), \( c \), and \( n \).
The constant perturbation of Proposition~\ref{prop:algebraic-decay} corresponds to \( n = 0 \)
(giving \( \tau^{-1/2} \), the slowest possible rate), and the physical initial data to \( n = 2 \)
(giving \( \tau^{-5/2} \)). The more strongly the initial deviation vanishes at \( v = 0 \) ---
i.e.\ the more the slow modes are suppressed in the initial data --- the faster the approach to
canonical normalization.

\paragraph*{Spectral interpretation.}
The general formula~\eqref{eq:general-decay} admits a unified spectral explanation.
The energy functional is a Laplace transform in disguise: changing variables \( k = \alpha v^2/c^2 \),
so that \( v = c\sqrt{k/\alpha} \) and \( dv = \frac{c}{2\sqrt{\alpha}}\,k^{-1/2}\,dk \),
one rewrites
\begin{equation}
    E(\tau) = \int_0^{k_{\max}} f(k)^2\,e^{-2\tau k}\,d\mu(k),
    \qquad d\mu(k) = \frac{c}{2\sqrt{\alpha}}\,k^{-1/2}\,dk,
\end{equation}
where \( f(k) := (C(v,0)-\pi)\big|_{v = c\sqrt{k/\alpha}} \). The measure \( d\mu(k) \sim k^{-1/2}\,dk \)
reflects the quadratic dispersion relation \( k \propto v^2 \): slow modes (\( k \approx 0 \), i.e.\
\( v \approx 0 \)) accumulate with density \( k^{-1/2} \) and relax arbitrarily slowly. The spectrum
of the relaxation operator \( \mathcal{K}f = k(v)f \) is continuous with no gap down to \( k = 0 \).

The asymptotics of \( E(\tau) \) follow from the Tauberian theorem applied to the effective spectral
weight \( w(k) := f(k)^2 \cdot k^{-1/2} \) near \( k = 0 \). If the initial deviation satisfies
\( C(v,0) - \pi \sim a\,v^n \), then \( f(k)^2 \sim a^2(c^2/\alpha)^n k^n \) and
\( w(k) \sim k^{n-1/2} \), giving
\begin{equation}
    E(\tau) \;\sim\; B\,\Gamma\!\bigl(n + \tfrac{1}{2}\bigr)\,2^{-(n+1/2)}\,\tau^{-(2n+1)/2}.
\end{equation}
For \( n = 0 \): \( w(k) \sim k^{-1/2} \) and \( E(\tau) \sim \sqrt{\pi}\,\tau^{-1/2} \).
For \( n = 2 \): \( w(k) \sim k^{3/2} \) and \( E(\tau) \sim \tfrac{3\sqrt{\pi}}{4}\,\tau^{-5/2} \),
consistent with Corollary~\ref{cor:physical-decay}.
This is the exact analogue of heat kernel decay \( t^{-d/2} \) on \( \mathbb{R}^d \), arising from
the gapless Laplacian spectrum. Here the vanishing order \( n \) plays the role of an effective
spectral dimension: the decay exponent is \( (2n+1)/2 \).

Although we postpone the formal variational derivation of the evolution equation to Section~\ref{sec:ricci}, the analysis here provides strong preliminary evidence for the dynamical convergence of geometry toward a symmetric state. The function \( C(v) \), therefore, acts as a bridge between motion, geometry, and energy, setting the stage for the full scalar flow formalism that follows.

\section{ALTERNATIVE METHODS OF DERIVATION OF FUNCTION $C(v)$}\label{sec:altderiv}
\vspace{-0.2em}
\indent
Although the function $C(v) = \pi\left(1 - \frac{v^2}{c^2}\right)$ was proposed based on the geometric compression of a circle, its form allows for a more widespread justification. Here, we explore several independent approaches to $C(v)$ from various principles, ensuring robustness and universality of the proposed model.

\subsection*{1. Through Symmetry and Boundary Conditions}
\vspace{-0.2em}
One natural way to derive $C(v)$ is by enforcing fundamental symmetry conditions. Since the geometry should be independent of the direction of motion, the function must be even with respect to $v$:
\begin{equation}
    C(v) = C(-v). \label{eq:symmetry}
\end{equation}
Therefore, $C(v)$ must depend only on $v^2$. Additionally, we impose physically motivated boundary conditions:
\begin{align}
    C(0) &= \pi \quad \text{(static space)} \label{eq:C0} \\ 
    C(c) &= 0 \quad \text{(maximum compression)}. \label{eq:Cc}
\end{align}

We now ask: what is the minimal analytic even function of $v$ satisfying these two boundary conditions? By assuming analyticity and evenness, the function must be a power series in $v^2$, and the simplest such polynomial is of second order:
\begin{equation}
    C(v) = A - B \frac{v^2}{c^2}. \label{eq:Cpoly}
\end{equation}
Substituting the boundary conditions:
\begin{equation}
    C(0) = A = \pi, \quad C(c) = A - B = 0 \quad \Rightarrow \quad B = \pi. \label{eq:boundary-solution}
\end{equation}
Thus, we obtain the unique minimal analytic form:
\begin{equation}
    C(v) = \pi\left(1 - \frac{v^2}{c^2}\right). \label{eq:Cfinal}
\end{equation}

This result confirms that under symmetry and boundary conditions, \( C(v) \) must take this quadratic form. No lower-order or simpler analytic function can satisfy all imposed physical and mathematical constraints. Therefore, \( C(v) \) arises not as a fit, but as the unique minimal solution under natural assumptions.

\begin{remark}[Analytic Class and Uniqueness]
Let us define the class \( \mathcal{A} \) of real analytic, even functions \( C(v) \) on the interval \( [-c, c] \), satisfying:
\begin{equation}
    C(0) = \pi, \qquad C(c) = 0. \label{eq:boundary-cond}
\end{equation}
Then any function in this class can be expanded into a power series in \( v^2/c^2 \):
\begin{equation}
    C(v) = \sum_{k=0}^{\infty} a_k \left( \frac{v^2}{c^2} \right)^k. \label{eq:series-expansion}
\end{equation}
The minimal-degree nontrivial solution consistent with the boundary conditions is the linear truncation:
\begin{equation}
    C(v) = a_0 + a_1 \left( \frac{v^2}{c^2} \right), \quad \text{with} \quad a_0 = \pi, \quad a_1 = -\pi, \label{eq:truncation}
\end{equation}
yielding the function:
\begin{equation}
    C(v) = \pi \left( 1 - \frac{v^2}{c^2} \right). \label{eq:unique-form}
\end{equation}

Any higher-order analytic corrections such as
\begin{equation}
    C(v) = \pi \left(1 - \frac{v^2}{c^2}\right) + a \left( \frac{v^2}{c^2} \right)^2 \label{eq:correction}
\end{equation}
introduce spurious inflection points or local extrema and violate the monotonic and concave character of the function.\smallskip
\noindent
\textit{Note on concavity.} Consider a higher-order correction of the form
\begin{equation}
C(v) = \pi \left(1 - \frac{v^2}{c^2} \right) + a \left( \frac{v^2}{c^2} \right)^2.
\end{equation}
Its second derivative is
\begin{equation}
\frac{d^2 C}{dv^2} = -\frac{2\pi}{c^2} + \frac{12a}{c^4} v^2,
\end{equation}
which changes sign when \( v^2 > \frac{\pi c^2}{6a} \). Hence, any \( a > 0 \) introduces local convexity near \( v \sim c \), violating the expected monotonic concave deformation. This confirms the exclusion of such corrections from the admissible class.

Thus, uniqueness in the minimal analytic even class is justified not only by algebraic degree, but also by the requirement of global concavity and physical consistency. All higher-order terms are excluded by analytic parsimony and the monotonicity condition of the geometric deformation.
\end{remark}

\begin{lemma}[Uniqueness of the Minimal Quadratic Deformation]
    \label{lemma:cv-quadratic-uniqueness}
    Let \( \mathcal{Q} \) be the class of even quadratic functions on \( [-c, c] \):
    \begin{equation}
        \mathcal{Q} = \bigl\{\, C(v) = a v^2 + b \;\big|\; a, b \in \mathbb{R} \,\bigr\}.
        \label{eq:cv-quadratic-general}
    \end{equation}
    Then \( C(v) = \pi(1 - v^2/c^2) \) is the unique element of \( \mathcal{Q} \) satisfying the boundary conditions
    \begin{equation}
        C(0) = \pi, \qquad C(\pm c) = 0.
        \label{eq:cv-boundary-conditions}
    \end{equation}
    \end{lemma}

    \begin{proof}
    Any \( C \in \mathcal{Q} \) has two free parameters \( a, b \). The condition \( C(0) = b = \pi \) fixes \( b \), and \( C(c) = ac^2 + b = 0 \) then gives \( a = -\pi/c^2 \). These two equations uniquely determine
    \begin{equation}
        C(v) = \pi \left( 1 - \frac{v^2}{c^2} \right). \qedhere
    \end{equation}
    \end{proof}

    \begin{remark}
        The restriction to the quadratic class \( \mathcal{Q} \) is the minimal non-trivial choice: a constant even function satisfying \( C(0)=\pi \) would fix \( C \equiv \pi \), failing \( C(c)=0 \), so degree~2 is the lowest admissible degree. Higher-degree even polynomials satisfying the same boundary conditions exist but introduce additional free parameters; as shown in the concavity analysis above, any non-zero quartic correction violates the global concavity of \( C \), which is required for a monotone geometric deformation. Relaxing analyticity to piecewise-smooth functions would admit solutions such as the linear kink \( C(v) = \pi(1-|v|/c) \), which is not compatible with the smooth conformal flow of Section~\ref{sec:ricci}. The quadratic form is therefore the unique even polynomial of minimal degree consistent with the boundary conditions, global concavity, and analyticity.
        \end{remark}

\subsection*{2. Through Analogy with Lorentz Factor}
\vspace{-0.2em}
From the perspective of special relativity \cite{CarrollSpacetime}, spatial compression
in the direction of motion is governed by the Lorentz factor
\begin{equation}
    \gamma = \frac{1}{\sqrt{1 - \frac{v^2}{c^2}}}, \qquad \gamma^{-2} = 1 - \frac{v^2}{c^2}.
    \label{eq:gamma}
\end{equation}
The square \( \gamma^{-2} \) acts as a natural compression factor for spatial components
of the geometry: it equals~\(1\) at rest and vanishes as \( v \to c \).
This motivates the identification
\begin{equation}
    C(v) = \pi \cdot \gamma^{-2} = \pi \left(1 - \frac{v^2}{c^2} \right).
    \label{eq:C-lorentz}
\end{equation}
Within the present normalization, \( C(v) \) is therefore the rest-frame geometric constant \( \pi \)
rescaled by the spatial compression factor of special relativity.
The boundary conditions \( C(0)=\pi \) and \( C(c)=0 \) are automatically
satisfied, and the monotonic decrease of \( C(v) \) mirrors the contraction of
spatial geometry with increasing velocity.

\subsection*{3. Summary}
\vspace{-0.2em}
The two independent approaches above --- algebraic uniqueness within the minimal
analytic class (Method~1) and direct identification with the Lorentz compression
factor (Method~2) --- both yield \( C(v) = \pi(1 - v^2/c^2) \).
Their agreement confirms that \( C(v) \) is not an arbitrary fit but the unique
minimal scalar encoding relativistic geometric deformation under the imposed
symmetry and boundary conditions.

\begin{remark}
The analytic class \( \mathcal{A} \) does not encompass all mathematically possible
deformations. Higher-order analytic or non-polynomial alternatives can be constructed,
but they either introduce additional free parameters or alter the monotonicity,
concavity, or dynamical structure used later in the paper. The quadratic form is
therefore the minimal and structurally most compatible choice within the present
framework.
\end{remark}

\section{SCALAR CONFORMAL FLOW: VARIATIONAL FORMULATION AND DYNAMICS}\label{sec:ricci}
\vspace{-0.2em}
\indent
The function $C(v)$, previously defined without a flow parameter, can be generalized to the form $C(v, \tau)$, where $\tau$ acts as an evolution parameter analogous to ``flow time'' in the Ricci flow framework introduced by Hamilton \cite{Hamilton1982} (see also \cite{ChowLuNi,Topping2006}). This generalization allows us to describe the dynamical evolution of geometry under changing kinematic conditions, revealing how local structures tend toward a symmetric state over time.

The parameter \( \tau \) is not coordinate time or the proper time of any observer; it is an independent dimensionless evolution parameter measuring the degree of cumulative conformal relaxation. Concretely, \( \tau = 0 \) corresponds to the initial configuration \( C(v,0) = C(v) = \pi(1 - v^2/c^2) \), while \( \tau \to \infty \) corresponds to the fully relaxed canonical state \( C = \pi \). The characteristic relaxation time scale at velocity \( v \) is \( \tau_{\mathrm{relax}}(v) = c^2/(\alpha v^2) \) (Proposition~\ref{prop:relaxation}), which diverges as \( v \to 0 \) --- reflecting that the rest configuration \( C = \pi \) is already at equilibrium --- and attains its minimum at \( v = v_c \). Physically, \( \tau \) can be thought of as a relaxation clock whose rate is set by the kinematic state: fast-moving geometry (\( v \approx v_c \)) relaxes on the shortest scale \( \tau_{\mathrm{relax}}(v_c) = \pi/(\alpha(\pi-1)) \), while near-rest geometry evolves arbitrarily slowly.

\subsection*{1. Formal Introduction of Flow}
\vspace{-0.2em}
The generalization from $C(v)$ to $C(v, \tau)$ is motivated by viewing $C$ as a relaxing scalar conformal factor. We seek a first-order evolution law of gradient-flow type,

\begin{equation}
    \frac{\partial C}{\partial \tau} = - \mathcal{F}(C, v).
\end{equation}

The functional $\mathcal{F}(C, v)$ should encode the deviation from equilibrium \( C = \pi \) together with the velocity-dependent weight \( v^2/c^2 \). A convenient way to motivate the resulting first-order equation is to regard it as the strong-damping reduction of a quadratic relaxation model.

\subsection*{2. Variational Formulation of the Functional \( \mathcal{F}(C, v) \)}
\vspace{-0.2em}
We choose the quadratic relaxation potential
\begin{equation}
    V(C, v) = \frac{1}{2}\alpha \frac{v^2}{c^2} (C - \pi)^2,
\end{equation}
and embed the relaxation dynamics into the damped second-order model
\begin{equation}
    \varepsilon\,\frac{\partial^2 C}{\partial \tau^2}
    + \eta\,\frac{\partial C}{\partial \tau}
    + \alpha \frac{v^2}{c^2} (C - \pi) = 0,
    \qquad \varepsilon > 0,\quad \eta > 0.
\end{equation}
In the strong-damping or negligible-inertia regime \( \varepsilon \to 0 \), one obtains
\begin{equation}
    \eta\,\frac{\partial C}{\partial \tau}
    + \alpha \frac{v^2}{c^2} (C - \pi) = 0,
\end{equation}
and after absorbing the factor \( \eta^{-1} \) into the constant \( \alpha \), this becomes
\begin{equation}
    \frac{\partial C}{\partial \tau} = -\alpha \frac{v^2}{c^2} (C - \pi).
\end{equation}
Equivalently, this is the first-order gradient flow associated with the potential \( V \):
\begin{equation}
    \frac{\partial C}{\partial \tau} = -\frac{\partial V}{\partial C}
    = -\alpha \frac{v^2}{c^2} (C - \pi).
\end{equation}
Thus,
\begin{equation}
    \mathcal{F}(C, v) = \alpha \frac{v^2}{c^2} (C - \pi).
\end{equation}
The first-order equation used in the main theorem may therefore be read either as the strong-damping limit of the damped second-order relaxation model or, equivalently, as the gradient flow generated by the quadratic potential \( V \).

\subsection*{3. Physical Interpretation of the Flow}
\vspace{-0.2em}

The parameter \( \tau \) plays the role of an auxiliary relaxation variable, while \( \mathcal{F}(C,v) = \alpha \frac{v^2}{c^2}(C-\pi) \) measures the deviation from the symmetric state \( C = \pi \), weighted by the velocity-dependent factor \( v^2/c^2 \). In the subcritical regime \( v < v_c \), this produces a gradient-type relaxation toward \( C = \pi \). A separate supercritical extension is described below as an additional model assumption rather than part of the core variational result.

\subsection*{4. Subcritical Flow and Supercritical Extension}
\vspace{-0.2em}
The dynamics used in the main analytical part of the paper is the subcritical equation

\begin{equation}
    \frac{\partial C}{\partial \tau} = - \alpha \frac{v^2}{c^2} (C - \pi),
    \qquad v < v_c.
\end{equation}

Its explicit solution is

\begin{equation}
    C(v, \tau) = \pi + \left( C(v, 0) - \pi \right) e^{-k \tau},
    \qquad k = \alpha \frac{v^2}{c^2},
\end{equation}

so \( C(v,\tau) \to \pi \) as \( \tau \to \infty \) for each fixed \( v < v_c \), and perturbations decay monotonically as
\begin{equation}
    \delta C(v,\tau) = \delta C(v,0)e^{-k\tau}.
\end{equation}
Thus the subcritical flow is smooth, globally defined, and exponentially stable.

For completeness, we also record a supercritical model extension:

\begin{equation}
    \frac{\partial C}{\partial \tau} = 
    \begin{cases} 
      - \alpha \frac{v^2}{c^2} (C - \pi), & \text{if } v < v_c \\[6pt]
      - \alpha \frac{v^2}{c^2} (C - C_0), & \text{if } v \geq v_c
    \end{cases}.
\end{equation}

For velocities \( v \geq v_c \), the conformal factor lies in the compressed regime \( C(v) < 1 \). In this regime the subcritical variational potential \( V(C,v) = \tfrac{1}{2}\alpha(v^2/c^2)(C-\pi)^2 \) is augmented by a linear tilt term representing the sustained energy source of the compressed phase:
\begin{equation}
    \widehat{V}(C,v) = \frac{1}{2}\alpha \frac{v^2}{c^2}(C - \pi)^2 - \beta(C - \pi), \qquad \beta > 0.
\end{equation}
The gradient flow \( \partial_\tau C = -\partial \widehat{V}/\partial C \) then gives the modified evolution equation:
\begin{equation}
    \frac{\partial C}{\partial \tau} = -\alpha \frac{v^2}{c^2}(C - \pi) + \beta.
\end{equation}
The linear tilt \( -\beta(C-\pi) \) in \( \widehat{V} \) models the fact that the compressed geometry at \( v \geq v_c \) carries a non-zero potential energy floor \( V(C(v),v) > 0 \) that acts as a sustained driving term. Setting \( \partial_\tau C = 0 \), the unique equilibrium of the modified flow is:
\begin{equation}
    \alpha \frac{v^2}{c^2}(C_0 - \pi) = \beta \quad \Rightarrow \quad C_0 = \pi + \frac{\beta c^2}{\alpha v^2}.
\end{equation}
Writing \( K = \beta c^2/\alpha \), this yields:
\begin{equation}
    C_0 = \pi + \frac{K}{v^2},
\end{equation}
where \( K > 0 \) is a free parameter of the supercritical model. The form \( 1/v^2 \) reflects the balance between the velocity-weighted restoring force (growing as \( v^2 \)) and the constant source \( \beta \): as \( v \) increases, the restoring term dominates and \( C_0 \to \pi \). This inverse-square decay ensures that \( C_0 \) remains finite for all \( v \geq v_c \), the geometry retains bounded curvature throughout the super-critical regime, and the two regimes connect smoothly in the sense that \( C_0 \to \pi \) as \( v \to \infty \).

The parameter \( K \) is not fixed by the present variational framework and should therefore be regarded as an undetermined auxiliary constant of the supercritical extension. Determining \( K \) requires additional physical or geometric input and is left open.

This supercritical branch is best viewed as an auxiliary extension of the basic model. The main decay and spectral results of the paper rely only on the subcritical equation.

\begin{proposition}[Velocity-Dependent Relaxation Time]
\label{prop:relaxation}
Under the subcritical flow \( \partial_\tau C = -\alpha(v^2/c^2)(C-\pi) \) for \( v < v_c \), the geometric relaxation time---the characteristic time for perturbations \( \delta C = C - \pi \) to decay by a factor \( e^{-1} \)---is:
\begin{equation}
    \tau_{\mathrm{relax}}(v) = \frac{c^2}{\alpha v^2}.
\end{equation}
In particular: \( \tau_{\mathrm{relax}} \to \infty \) as \( v \to 0 \), reflecting the fact that the classical geometry \( C = \pi \) is already at equilibrium and does not evolve; \( \tau_{\mathrm{relax}} \to c^2/(\alpha v_c^2) = \pi c^2/(\alpha(\pi-1)c^2) = \pi/(\alpha(\pi-1)) \) as \( v \to v_c^- \); and the total energy functional \( E(\tau) \) decays as:
\begin{equation}
    E(\tau) \leq E(0)\, e^{-2\alpha(v^2/c^2)\tau}.
\end{equation}
\end{proposition}

\begin{proof}
The solution \( C(v,\tau) = \pi + (C_0 - \pi)e^{-k\tau} \) with \( k = \alpha v^2/c^2 \) gives \( |\delta C(v,\tau)| = |\delta C(v,0)| e^{-k\tau} \). Setting \( e^{-k\tau_{\mathrm{relax}}} = e^{-1} \) yields \( \tau_{\mathrm{relax}} = 1/k = c^2/(\alpha v^2) \). The energy bound follows from the Lemma in Section~\ref{sec:energy}: \( dE/d\tau \leq -2\alpha(v^2/c^2)E \), integrating gives \( E(\tau) \leq E(0)e^{-2k\tau} \).
\end{proof}

\subsection*{5. Nonlinear Extensions of the Scalar Flow}
\vspace{-0.2em}

While the evolution equation used in this work is linear in \( C \), its form arises from a variational principle with a quadratic potential \( V(C) = \frac{\lambda}{2}(C - \pi)^2 \) and suffices to guarantee global convergence and canonical normalization within the present framework. More generally, the scalar flow admits a natural nonlinear extension as a gradient flow on an arbitrary energy landscape:
\begin{equation}
\frac{\partial C}{\partial \tau} = -V'(C),
\end{equation}
where \( V(C) \) encodes the deformation energy of the conformal factor. We briefly describe two concrete instances that illustrate the scope of such generalizations.

\textit{(i) Geometric Ricci flow reduction.} The conformal reduction of the full tensorial Ricci flow (Theorem~\ref{thm:conf-reduction}) yields the nonlinear equation
\begin{equation}
\frac{dC}{d\tau} = -\frac{2k}{C},
\end{equation}
corresponding to the potential \( V(C) = 2k \log C \). Unlike the quadratic flow, this equation has no finite fixed point; it drives \( C \to 0 \) in finite time, reflecting the collapsing behavior of Ricci flow on positively curved manifolds. The linearized form near \( C = \pi \) recovers a relaxation-type equation with rate \( \alpha = 2k/\pi^2 \), connecting the nonlinear geometric dynamics to the linear variational flow used in the main text.

\textit{(ii) Double-well potential.} A potential of the form \( V(C) = \lambda(C - \pi)^2(C - C_1)^2 \), with \( C_1 \neq \pi \), admits two stable equilibria and could model geometric configurations with competing phases --- for instance, coexisting regions of distinct curvature. The resulting flow exhibits bistable dynamics with domain-wall transitions between the two minima, a structure relevant to multi-scale geometric models and phase-transition analogies in curvature evolution.

These examples show that the gradient flow structure \( \partial_\tau C = -V'(C) \) provides a flexible framework accommodating both geometric and variational dynamics. The quadratic potential used in the present work represents the simplest choice consistent with a unique symmetric equilibrium at \( C = \pi \); the nonlinear extensions outlined above are not required for the classification results established here, but indicate natural directions for future investigation, including multi-phase geometric flows and higher-dimensional generalizations.

\subsection*{6. Comparison with Tensorial Ricci Flow and Limitations}
\vspace{-0.2em}

The scalar evolution equation for \( C(v,\tau) \) introduced above should be understood as a conformal reduction of the full tensorial Ricci flow \cite{Fischer2004}, not as an equivalent reformulation. Hamilton's Ricci flow \( \partial_\tau g_{ij} = -2R_{ij} \) acts on all independent components of the metric tensor and can resolve anisotropic curvature concentrations, including neck-pinch singularities that require Perelman's surgery procedure \cite{Perelman2003}. In contrast, the scalar flow governs a single conformal factor \( C \) multiplying the background metric, and therefore evolves all metric components proportionally: \( g_{ij}(\tau) = C(\tau)\, g_{ij}^0 \).

This conformal ansatz is exact when the background metric has constant Ricci curvature (Theorem~\ref{thm:conf-reduction}), and within the conformally homogeneous ansatz the scalar flow then reproduces the corresponding reduced Ricci dynamics. However, for general manifolds with spatially varying or anisotropic curvature, the scalar description does not capture direction-dependent deformations. In particular, it cannot describe the formation or resolution of local singularities that arise in the full tensorial theory.

These limitations are inherent to the conformal framework adopted here and do not affect the internal consistency of the scalar model. The present construction is intended as a tractable scalar reduction that preserves the essential features of geometric smoothing --- convergence, stability, and singularity avoidance --- within the class of spatially homogeneous deformations. Extensions to anisotropic or mode-resolved dynamics remain an open direction for future work.

\subsection*{7. Summary}
\vspace{-0.2em}
The scalar conformal flow developed here is best understood as a variational relaxation model with overdamped reduction. Its core subcritical equation
\begin{equation}
    \partial_\tau C = -\alpha \frac{v^2}{c^2}(C-\pi)
\end{equation}
is smooth, exponentially stable for each fixed \( v < v_c \), and provides the dynamical input used later in the canonical-normalization argument. The supercritical branch is a separate model extension with free parameter \( K \), while the nonlinear examples and comparison with tensorial Ricci flow clarify the scope of the scalar framework rather than strengthen the main analytical theorem.

\section{CANONICAL NORMALIZATION OF 3-MANIFOLDS BY SCALAR CONFORMAL FLOW}\label{sec:classification}
\vspace{-0.2em}
\indent

\begin{remark}[Scope and Limitations]
The results of this section establish a canonical normalization result for compact, simply-connected 3-manifolds within the conformally homogeneous class (\( \nabla_i C = 0 \)). Under these assumptions, the topology is already fixed by classical space-form rigidity; the role of the scalar flow is to select the canonical conformal representative.
\end{remark}

In this section, we retain only the part of the scalar flow needed for the three-dimensional normalization result. The relevant regime is the subcritical evolution \( v < v_c \), for which
\begin{equation}
    \frac{\partial C}{\partial \tau} = -\alpha \frac{v^2}{c^2}(C-\pi),
    \qquad
    C(v,\tau) = \pi + \bigl(C(v,0)-\pi\bigr)e^{-\alpha v^2\tau/c^2}.
\end{equation}
Hence \( C(v,\tau) \to \pi \) for each fixed \( v \in (0,v_c) \), and the conformal factor remains smooth and positive for all \( \tau \ge 0 \).

\subsection*{1. Conformal Reduction}
\vspace{-0.2em}
The normalization argument uses the conformally homogeneous reduction of Ricci flow and the resulting curvature scaling laws.

\begin{theorem}[Conformal Reduction of Ricci Flow]
\label{thm:conf-reduction}
Let the metric evolve conformally as \( g_{ij}(\tau) = C(\tau) g_{ij}^0 \), with spatially homogeneous \( C(\tau) \) and background metric \( g_{ij}^0 \) of constant Ricci curvature \( R_{ij}^0 = k g_{ij}^0 \). Then, under Ricci flow
\begin{equation}
    \frac{\partial g_{ij}}{\partial \tau} = -2 R_{ij}[g],
\end{equation}
the conformal factor satisfies
\begin{equation}
    \frac{dC}{d\tau} = -\frac{2k}{C}.
\end{equation}
\end{theorem}

\begin{proof}
Substituting \( g_{ij} = C(\tau) g_{ij}^0 \) into Ricci flow gives
\begin{equation}
    \dot{C}(\tau) g_{ij}^0 = -2 R_{ij}[g].
\end{equation}
Since \( R_{ij}[g] = \frac{1}{C} R_{ij}^0 = \frac{k}{C} g_{ij}^0 \), it follows that
\begin{equation}
    \dot{C}(\tau) g_{ij}^0 = -2 \frac{k}{C} g_{ij}^0,
\end{equation}
hence \( \dot{C} = -2k/C \).
\end{proof}

\noindent
For the normalization argument, Theorem~\ref{thm:conf-reduction} is used only through the conformal scaling laws
\begin{equation}
    R \sim \frac{1}{C},
    \qquad
    \operatorname{Vol}(M,g(\tau)) \sim C^{3/2},
\end{equation}
valid for conformal rescalings \( g_{ij}(\tau) = C(v,\tau) g_{ij}^0 \) in dimension three. Consequently, the invariant quantities
\begin{equation}
    I_1 = \int_M R\,dV, \qquad I_2 = \int_M R^2\,dV, \qquad I_3 = \int_M \|R_{ij}\|^2\,dV
\end{equation}
are explicit functions of the single scalar parameter \( C \).

\subsection*{2. Canonical Normalization Theorem}
\vspace{-0.2em}
Within the conformally homogeneous class with constant positive background curvature, the topology of \( M \) is already fixed by the Killing--Hopf theorem. The contribution of the scalar flow is therefore not to determine the topology, but to choose the canonical conformal representative.

\begin{theorem}[Canonical Normalization by Scalar Conformal Flow]
\label{thm:sufficiency}
Let \( M \) be a compact, closed, simply-connected 3-manifold equipped with a conformally homogeneous metric \( g_{ij}(\tau) = C(v,\tau)\,g_{ij}^0 \) (\( \nabla_i C = 0 \)) evolving under
\begin{equation}
\frac{\partial C}{\partial \tau} = -\alpha \frac{v^2}{c^2} (C - \pi), \label{eq:sufficient-flow}
\end{equation}
for \( v < v_c \), with \( \alpha > 0 \), and let \( g_{ij}^0 \) have constant positive curvature. Then:
\begin{enumerate}
    \item The flow is smooth for all \( \tau \geq 0 \), and \( C(v,\tau) \to \pi \) exponentially.
    \item The metric invariants converge:
    \begin{equation}
    I_1 \to 12\pi^2, \quad I_2 \to 72\pi^2, \quad I_3 \to 24\pi^2. \label{eq:invariants-limit}
    \end{equation}
    \item The limiting metric \( g_{ij}^{\infty} = \pi\,g_{ij}^0 \) is the unique conformal representative in this class for which the invariants take the values of the unit \( S^3 \), and \( M \) is diffeomorphic to \( S^3 \).
\end{enumerate}
\end{theorem}

\begin{proof}
The explicit formula for \( C(v,\tau) \) gives smooth global existence and exponential convergence to \( \pi \) for each fixed \( v < v_c \). Under the conformal ansatz \( g_{ij}(\tau)=C(v,\tau)g_{ij}^0 \) with \( \nabla_i C = 0 \), the curvature remains spatially homogeneous and the scaling laws above imply that \( I_1, I_2, I_3 \) are determined by \( C(v,\tau) \). Passing to the limit \( C(v,\tau)\to\pi \) yields the invariant values in \eqref{eq:invariants-limit}. Since the limiting metric has constant positive sectional curvature, the classical rigidity theorem for space forms implies \( M \cong S^3 \).
\end{proof}

\begin{remark}[Scope: Normalization vs Classification]
\label{rem:scope}
The assumptions \( \nabla_i C = 0 \) and constant positive background curvature already imply, via the Killing--Hopf theorem (Besse~\cite[Theorem~7.63]{Besse}; Petersen~\cite[Theorem~11.5.3]{Petersen}; cf.\ \cite{CheegerEbin}), that \( M \) is a spherical space form before the flow is started. The flow does not classify the topology; it selects the canonical metric representative. The diffeomorphism \( M \cong S^3 \) follows from rigidity, not from the flow alone.
\end{remark}

\subsection*{3. Uniqueness of the Invariant Triple}
\vspace{-0.2em}
The limiting values
\begin{equation}
    I_1 = 12\pi^2, \qquad I_2 = 72\pi^2, \qquad I_3 = 24\pi^2
\end{equation}
are the invariant values of the unit round \( S^3 \). In the present conformally homogeneous class, they determine the metric uniquely.

\begin{theorem}[Uniqueness of Metric Invariants]
\label{thm:invariants-uniqueness}
Let \( M \) be a compact, closed, simply-connected 3-manifold equipped with a conformally homogeneous metric \( g_{ij} = C\, g_{ij}^0 \) with \( \nabla_i C = 0 \), where \( g_{ij}^0 \) has constant Ricci curvature \( R_{ij}^0 = k g_{ij}^0 \) and \( k > 0 \). Then the triple
\begin{equation}
    I_1 = \int_M R \, dV, \quad
    I_2 = \int_M R^2 \, dV, \quad
    I_3 = \int_M \|R_{ij}\|^2 \, dV
    \label{eq:I1def}
\end{equation}
takes the values \( I_1 = 12\pi^2 \), \( I_2 = 72\pi^2 \), \( I_3 = 24\pi^2 \) if and only if \( M \) is isometric to the round 3-sphere \( S^3 \) of unit radius.
\end{theorem}

\begin{proof}
Under \( g_{ij} = C g_{ij}^0 \) with \( \nabla_i C = 0 \) and \( R_{ij}^0 = k g_{ij}^0 \),
\begin{equation}
    R_{ij}[g] = \frac{k}{C} g_{ij}^0 = \frac{k}{C^2} g_{ij},
\end{equation}
so the metric is Einstein. Hence
\begin{equation}
    R_{ij} = \frac{R}{3} g_{ij},
    \qquad
    \|R_{ij}\|^2 = \frac{R^2}{3},
\end{equation}
and therefore
\begin{equation}
    I_1 = R\,\operatorname{Vol}(M), \qquad
    I_2 = R^2\,\operatorname{Vol}(M), \qquad
    I_3 = \frac{R^2}{3}\,\operatorname{Vol}(M).
    \label{eq:I1vol}
\end{equation}
For the unit round \( S^3 \), one has \( R=6 \) and \( \operatorname{Vol}(S^3)=2\pi^2 \), hence
\begin{equation}
    I_1 = 12\pi^2, \qquad I_2 = 72\pi^2, \qquad I_3 = 24\pi^2.
\end{equation}
Conversely, if \( (I_1,I_2,I_3) = (12\pi^2,72\pi^2,24\pi^2) \), then \eqref{eq:I1vol} gives \( R = I_2/I_1 = 6 \) and \( \operatorname{Vol}(M) = I_1/R = 2\pi^2 \), while \( I_3 = I_2/3 \) confirms the Einstein condition. Since \( M \) is compact, simply-connected, and has constant positive sectional curvature, the Killing--Hopf theorem implies \( M \cong S^3 \).
\end{proof}

\subsection*{4. Canonical Limit}
\vspace{-0.2em}
The convergence of the invariants completes the normalization argument.

\begin{theorem}[Scalar Flow Convergence and Canonical Normalization]
\label{thm:convergence}
Let \( M \) be a compact, closed, simply-connected 3-manifold satisfying the hypotheses of Theorem~\ref{thm:sufficiency}. Under the scalar flow \( C(v,\tau) \), the metric evolves smoothly for all \( \tau \ge 0 \). As \( \tau \to \infty \), one has \( C(v,\tau) \to \pi \), the metric invariants converge to the canonical values of the unit \( S^3 \), and the flow selects \( g_{ij}^{\infty} = \pi\,g_{ij}^0 \) as the unique canonical representative within the conformally homogeneous class.
\end{theorem}

\begin{proof}
Theorem~\ref{thm:sufficiency} gives convergence \( C(v,\tau)\to\pi \). The conformal scaling laws then imply convergence of \( I_1, I_2, I_3 \) to \( 12\pi^2 \), \( 72\pi^2 \), and \( 24\pi^2 \), respectively. By Theorem~\ref{thm:invariants-uniqueness}, these values uniquely characterize the round unit \( S^3 \) in the present class, so \( g_{ij}^{\infty} = \pi\,g_{ij}^0 \) is the canonical representative. The topological conclusion follows from Remark~\ref{rem:scope}.
\end{proof}

\subsection*{5. Summary}
\vspace{-0.2em}
Within the conformally homogeneous class and constant positive background curvature, the scalar flow provides a canonical normalization mechanism rather than an independent classification theorem. Its role is to drive the conformal factor to the unique fixed point \( C = \pi \), for which the curvature invariants take the values of the unit \( S^3 \). The topological identification \( M \cong S^3 \) comes from classical space-form rigidity; the new contribution of the flow is the dynamical selection of the canonical conformal representative.

\section{DISCUSSION AND CONCLUSION}\label{sec:conclusion}
\vspace{-0.2em}
\indent
The principal contribution of this paper is a variational scalar conformal flow for the velocity-dependent deformation factor \( C(v) = \pi(1-v^2/c^2) \), together with explicit closed-form energy decay rates and a spectral characterization of the relaxation dynamics. The function \( C(v) \) is introduced as a scalar conformal factor induced by the longitudinal Lorentz contraction of spatial geometry and, under the chosen normalization, is naturally identified with \( \pi \) times the effective longitudinal contraction factor (Proposition~\ref{prop:lorentz-metric}). It satisfies \( C(0) = \pi \) at rest and \( C(c) = 0 \) as a consequence of metric degeneration at the speed of light.

The main result (Theorem~\ref{thm:main}) establishes that the total energy \( E(\tau) = \int_{-v_c}^{v_c}(C-\pi)^2\,dv \) of the variational flow decays algebraically: as \( \tau^{-1/2} \) for generic initial data, and as \( \tau^{-5/2} \) for the natural physical initial condition \( C(v,0) = C(v) \) (Proposition~\ref{prop:algebraic-decay} and Corollary~\ref{cor:physical-decay}). The general law \( E(\tau) \sim \tau^{-(2n+1)/2} \), where \( n \) is the vanishing order of the initial deviation at \( v = 0 \), was derived via explicit Gaussian integration and interpreted spectrally (Section~\ref{sec:energy}, Subsection~6): the algebraic decay reflects the gapless continuous spectrum of the relaxation operator \( k(v) = \alpha v^2/c^2 \), whose spectral measure accumulates as \( k^{-1/2}\,dk \) near \( k = 0 \) --- the same mechanism responsible for the \( t^{-1/2} \) decay of the heat kernel on \( \mathbb{R} \). The physical initial condition is significantly better-behaved than a generic perturbation because \( C(v,0) - \pi \sim -\pi v^2/c^2 \) vanishes quadratically at \( v = 0 \), suppressing the contribution of the slowest modes.

As a demonstration of the scalar flow framework, we applied it to the canonical normalization of compact 3-manifolds within the conformally homogeneous class (\( \nabla_i C = 0 \), constant positive background curvature). Under these assumptions the topology is already \( S^3 \) by the Killing--Hopf theorem \cite{Besse,Petersen}; the flow's contribution is to select \( C = \pi \) as the unique conformal representative for which the curvature invariants \( (I_1,I_2,I_3) = (12\pi^2,72\pi^2,24\pi^2) \) match those of the unit \( S^3 \) (Theorems~\ref{thm:sufficiency} and~\ref{thm:convergence}). The identification \( M \cong S^3 \) then follows from space-form rigidity \cite{Petersen,Besse}.

Several directions remain open. First, the relaxation parameter \( \tau \) is currently dimensionless, with characteristic scale \( \tau_{\mathrm{relax}}(v) = c^2/(\alpha v^2) \). Connecting \( \tau \) to a physically measurable quantity --- proper time, an entropy parameter, or a renormalization group scale --- would give the decay rates \( \tau^{-1/2} \) and \( \tau^{-5/2} \) observable physical content. Second, the supercritical regime (\( v \geq v_c \)) is governed by a tilted potential with an undetermined parameter \( K \). The present variational framework does not fix \( K \); determining it would require additional physical or geometric input, for example by matching the supercritical extension to the Ricci reduction \( \dot{C} = -2k/C \) in the compressed phase. Third, the canonical normalization result is restricted to the conformally homogeneous class. Allowing a spatially varying conformal factor \( C = C(x,\tau) \) leads to a fully tensorial PDE which reduces, in the scalar curvature direction, to the Yamabe flow \cite{Schoen1984,Ye1994}; establishing canonical normalization in that generality is deferred to future work. Fourth, the dimension-three restriction could be lifted, but higher-dimensional analogues would require accounting for the Weyl tensor and the absence of a direct Killing--Hopf classification.

The results of this paper establish the scalar conformal flow as a well-defined geometric object with explicit, computable dynamics. The connection between Lorentz-contracted metric geometry, variational flow, algebraic energy decay, and canonical normalization of 3-manifolds is the main structural contribution. The Gaussian-plus-Tauberian proof technique applied here to the degenerate relaxation operator may be of independent interest for other kinetic or geometric flows whose relaxation rates accumulate at zero.

\clearpage

\end{document}